\documentclass{amsart}

\usepackage{amsmath,amssymb,lineno,amsthm,fullpage,parskip,graphicx}

\newtheorem{theorem}{Theorem}
\newtheorem{lemma}[theorem]{Lemma}
\newtheorem{corollary}[theorem]{Corollary}

\title{A Note on Colourings of Connected $2$-edge Coloured Cubic Graphs}

\begin{document}

\maketitle
\author{Christopher Duffy}\\
\address{Department of Mathematics and Statistics, University of Saskatchewan, CANADA}
\begin{abstract}
In this short note we show that every connected $2$-edge coloured cubic graph admits an $10$-colouring. This lowers the best known upper bound for the chromatic number of connected $2$-edge coloured cubic graphs.
\end{abstract}

A \emph{$2$-edge coloured graph} is a pair $(G,s)$ where $G$ is a simple graph and $s:E(G) \to \{+,-\}$.
When there is no chance for confusion we refer to the $2$-edge coloured graph  $(G,s)$ as $G$.
Alternatively, in various portions of the literature (\cite{NRS15}, \cite{S14} for example), $2$-edge coloured graphs are called \emph{signified graphs}.
We call an edge $e\in E(G)$ \emph{positive} when $s(e) = +$.
Otherwise we say $e$ is \emph{negative}.
We call a vertex $v \in (G,s)$  \emph{positive} when $s(uv) = +$ for all $u \in N(v)$.
We call a vertex $y \in (G,s)$ \emph{negative} when $s(xy) = -$ for all $x \in N(v)$.

Let $(G,s)$ and $(H,t)$ be $2$-edge-coloured graphs.
There is a \emph{homomorphism of $(G,s)$ to $(H,t)$}, denoted $(G,s) \to (H,t)$  when there exists a homomorphism $\phi: G \to H$ so that for all $uv \in E(G)$ we have $s(uv) = t(\phi(u)\phi(v))$. 
When $\phi$ is such a homomorphism we write $\phi: (G,s) \to (H,t)$.
Informally, a homomorphism of $(G,s)$ to $(H,t)$ is a homomorphism of $G$ to $H$ that preserves edge colours.
Following the conventions from study of graph colouring and homomorphism, when $\phi$ is a homomorphism of $(G,s)$ to $(H,d)$ we say that $\phi$ is a $|H|$-colouring of $(G,s)$.
We let $\chi\left((G,s)\right)$ denote the least integer $k$ such that $(G,s)$ admits a $k$-colouring. 
We call $\chi((G,s))$ the \emph{chromatic number} of $(G,s)$.
If $\mathcal{F}$ is a family of $2$-edge coloured graphs, $\chi(\mathcal{F})$ is defined to be the least integer $k$ such that $\chi(F) \leq k$ for all $F \in \mathcal{F}$.

Equivalently one may define a \emph{$k$-colouring} of $(G,s)$  as a labelling $c:V(G) \to \{0,1,\dots, k-1\}$ such that:
\begin{enumerate}
	\item $c(u) \neq c(v)$ for all $uv \in E(G)$;
	\item if $uv, xy \in E(G)$ with $c(u) = c(x)$ and $c(v) = c(y)$ , then $s(uv) = s(xy)$.
\end{enumerate}

Such a labeling implicitly defines a homomorphism to a $2$-edge coloured graph $(H,t)$ with vertex set $\{0,1,\dots, k-1\}$ and $t(ij) = \star$ ($\star \in \{+,-\}$) when there is an edge $uv \in E(G)$ such that $c(u)=i$ and $c(v) = j$ and $s(uv) = \star$.

We refer the reader to \cite{moprs08} for a comprehensive survey on homomorphisms and colourings of $2$-edge coloured graphs and their relation to \emph{signed graphs}.
We call a graph \emph{properly subcubic} when it has maximum degree at most three and contains a vertex of degree at most two.
For graph theoretic notation and terminology not defined here, we refer the reader to \cite{bondy2008graph}.

Let $SP_9$ be the $2$-edge coloured graph $(K_9,s)$ given in Figure \ref{fig:SP9}.
This $2$-edge coloured graph appears in previous studies of $2$-edge coloured graphs (for example \cite{D19Grids,moprs08,O14}) due to the following convenient adjacency properties.
Observe that $SP_9$ is isomorphic to the $2$-edge coloured graph formed from $SP_9$ by changing the colour of every edge as $K_3\square K_3$ is isomorphic to its own complement.
\begin{figure}
	\begin{center}
		\includegraphics[scale=.5]{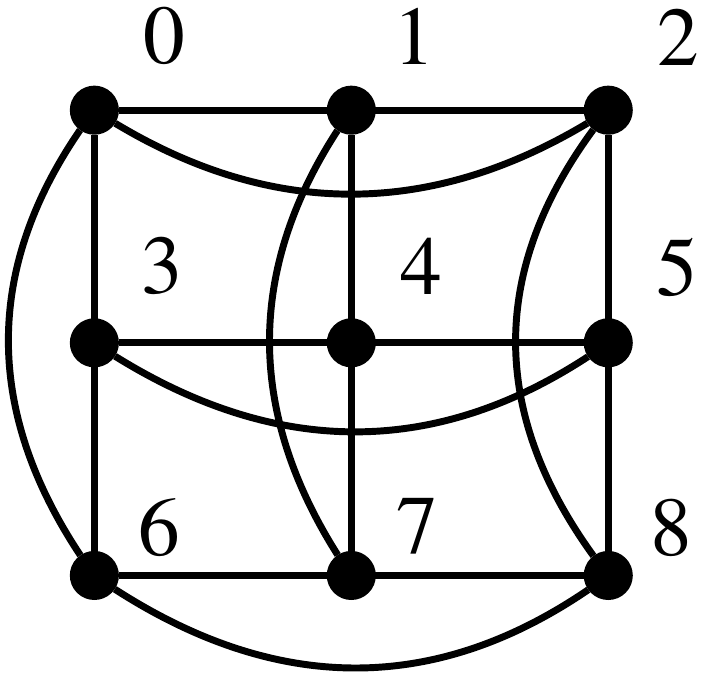}
	\end{center}
	\caption{$SP_9$ -- negative edges are not shown}
	\label{fig:SP9}
\end{figure}

\begin{lemma}\label{lem:SP9adj}
	For every positive edge $uv$ of $SP_9$ there exists
	\begin{itemize}
		\item a unique vertex $z$ such that $uz$ and $vz$ are positive;
		\item exactly two vertices, $z_1,z_2$ such that $uz_1$ and $uz_2$ are positive and $vz_1$ and $vz_2$ are negative;
		\item exactly two vertices, $z_1,z_2$ such that $uz_1$ and $uz_2$ are negative and $vz_1$ and $vz_2$ are positive;
		\item exactly two vertices  $z_1,z_2$ such that $uz_1$, $uz_2$, $vz_1$ and $vz_2$ are negative.
	\end{itemize}
	For every negative edge $xy$ of $SP_9$ there exists
\begin{itemize}
	\item  a unique vertex $z$ such that $uz$ and $vz$ are negative;
	\item exactly two vertices, $z_1,z_2$ such that $uz_1$ and $uz_2$ are negative and $vz_1$ and $vz_2$ are positive;
	\item exactly two vertices, $z_1,z_2$ such that $uz_1$ and $uz_2$ are positive and $vz_1$ and $vz_2$ are negative;
	\item exactly two vertices  $z_1,z_2$ such that $uz_1$, $uz_2$, $vz_1$ and $vz_2$ are positive.
\end{itemize}
\end{lemma}

Lemma \ref{lem:SP9adj} is easily confirmed by inspection. Lemma \ref{lem:SP9adj} is closely related to the notion of property $P_{i,j}$ for oriented graphs (see \cite{SV96}).  

Recall the following results for the colourings and homomorphisms of $2$-edge coloured cubic graphs.
Let $K_4^{s+}$ and $K_4^{s-}$ be the $2$-edge coloured cubic graphs given in Figure \ref{fig:k4s} (dashed edges indicate negative edges).
The $2$-edge coloured graph $SP_9^\star$ is constructed from $SP_9$ by adding two additional vertices $0^\prime$ and $1^\prime$ so that
$s(0^\prime0) = s(0^\prime1^\prime) = - $,  $s(1^\prime1)=+$ and $s(0^\prime k) = s(0 k)$, $s(1^\prime k) = s(1 k)$ for all $k \in \{2,3,4,5,6,7\}$.

\begin{figure}
	\begin{center}
	\includegraphics[scale=.5]{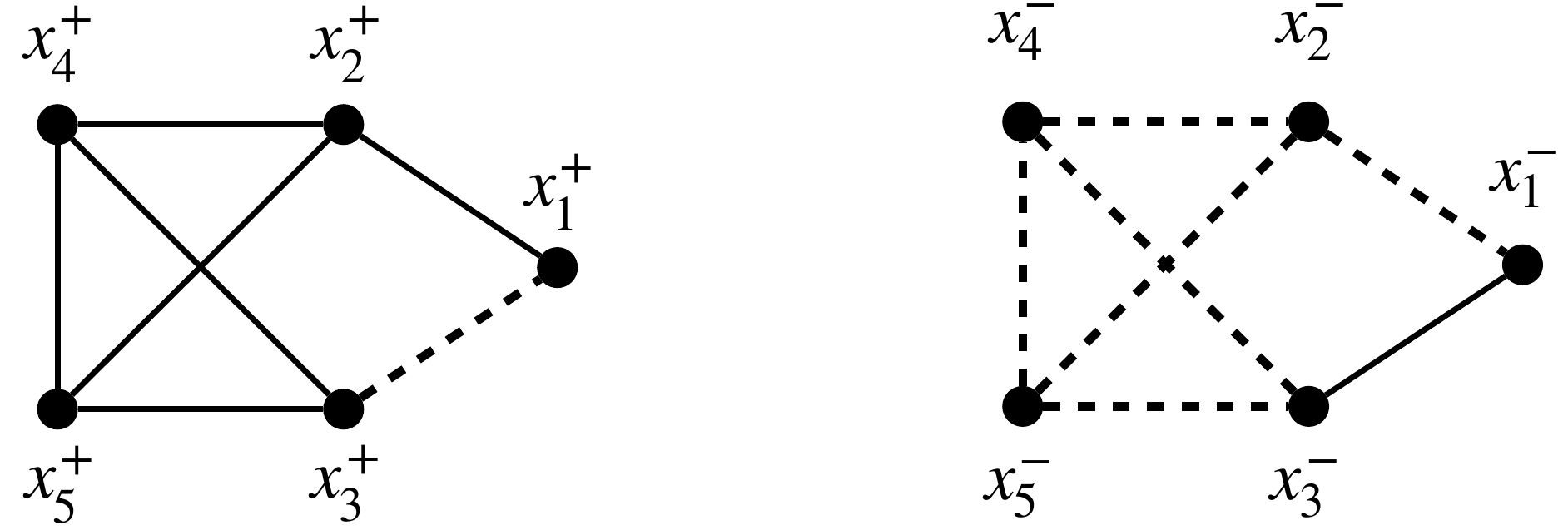}
	\end{center}
	\caption{$K_4^s+$ and $K_4^s-$}
	\label{fig:k4s}
\end{figure}

\begin{lemma}\label{lem:subCubic}\cite{D15Thesis,J19}
If $G$ is a properly subcubic $2$-edge coloured graph that contains no induced copy of $K_4^{s+}$ or $K_4^{s-}$, then $G \to SP_9$.
\end{lemma}

\begin{theorem}\label{thm:PinlouEtal.}\cite{J19}
	If $G$ is a $2$-edge coloured cubic graph, then $G \to SP_9^\star$.
\end{theorem}

\begin{theorem} \label{thm:familyBound}
	For $\mathcal{G}_3$, the family of $2$-edge coloured cubic graphs,  $8 \leq \chi(\mathcal{G}_3) \leq 11$.
\end{theorem}	

The lower bound in Theorem \ref{thm:familyBound} comes by way exhibition of a $2$-edge coloured graph with chromatic number $8$ (see \cite{D15Thesis} or \cite{J19}).
The upper bound first appeared in \cite{D15Thesis}, but restricted to connected $2$-edge coloured cubic graphs,
This upper bound was later generalized for all $2$-edge coloured cubic graphs in \cite{J19}.
This upper bound is not known to be tight.
In this note we improve this upper bound for the class of connected $2$-edge coloured cubic graphs.

We begin with two technical lemmas. Let $SP_9^\dagger$ be the $2$-edge coloured graph formed from $SP_9$ by adding a new vertex $z$ so that there is a positive edge $zu$ for all $u \in \{0,1,2\}$ and a negative edge $zv$ for all $v \in \{6,7,8\}$.

\begin{lemma}\label{lem:extendPlus}
	Let $G$ be a $2$-edge coloured graph formed by appending a vertex $v$ and an edge $vx_1^+$ to a copy of $K_4^{s+}$.
	For every $i \in V(SP_9)$, there exists a homomorphism $\phi: G \to SP_9^\dagger$ so that $\phi(v) = i$ and $\phi(x_5^+) = z$.
\end{lemma}

\begin{proof}
	We extend a partial homomorphism $\phi: G[\{v\}] \to SP_9^\dagger$ where $\phi(v) = i$.
	
	Consider the case $i = 0$. 
	If $vx_1^+$ is positive, then we can extend $\phi$ to include $x_1$ so that $\phi(x_1) \in \{1,2, 3, 6\}$.
	Notice that each vertex of $SP_9$ is adjacent to some vertex of $\{1,2, 3, 6\}$ by a positive edge.
	Therefore when $vx_1^+$ is positive we may extend $\phi$ so that $\phi(x_2) \in \{0,1,2,3,4,5,6,7,8\}$.
	Similarly note that every vertex of $SP_9$ other than $0$ is adjacent to some vertex of $\{1,2, 3, 6\}$ by a negative edge.
	Therefore when $vx_1^+$ is positive we may extend $\phi$ so that $\phi(x_3) \in \{1,2,3,4,5,6,7,8\}$.
	In particular we can extend $\phi$ so that $\phi(x_2) \cup \phi(x_3) = \{0,1\}$.
	We now extend to the remaining  vertices by letting $\phi(x_4^+) = 2$ and $\phi(x_5^+) = z$.
	
	If $vx_1^+$ is negative, then we can extend $\phi$ to include $x_1$ so that $\phi(x_1) \in \{4,5, 7, 8\}$.
	Notice that each vertex of $SP_9$ is adjacent to some vertex of $\{4,5, 7, 8\}$ by a negative edge.
	Therefore when $vx_1^+$ is negative we may extend $\phi$ so that $\phi(x_3) \in \{0,1,2,3,4,5,6,7,8\}$.
	Similarly note that every vertex of $SP_9$ other than $0$ is adjacent to some vertex of $\{4,5, 7, 8\}$ by a positive edge.
	Therefore when $vx_1^+$ is negative we may extend $\phi$ so that $\phi(x_2) \in \{1,2,3,4,5,6,7,8\}$.
	In particular we can extend $\phi$ so that $\phi(x_2) \cup \phi(x_3) = \{0,1\}$.
	We now extend to the remaining vertices by letting $\phi(x_4^+) = 2$ and $\phi(x_5^+) = z$.
	
	Consider the case $i \neq 0$. 
	Since $SP_9$ is vertex transitive, the above argument implies that regardless of the sign of $vx_1^+$ we may extend $\phi$ so that $\phi(x_2^+) \cup \phi(x_3^+) = \{0,1\}$.
 	Thus we may extend $\phi$ so that $\phi(x_4^+) = 2$ and $\phi(x_5^+) = z$.
 	This completes the proof.
\end{proof}

\begin{lemma} \label{lem:extendNeg}
	Let $G$ be a $2$-edge coloured graph formed by appending a vertex $v$ and an edge $vx_1^-$ to a copy of $K_4^{s-}$.
	For every $i \in V(SP_9)$ there exists a homomorphism $\phi: G \to SP_9^\dagger$ so that $\phi(v) = i$ and $\phi(x_5^-) = z$.
\end{lemma}

\begin{proof}
	The proof this result follows similarly to previous lemma by noting that a partial homomorphism $\phi: G[\{v\}  ] \to SP_9^\dagger$ with $\phi(v) = i$ can be extended so that  $\phi(x^-_2) \cup \phi(x^-_3) = \{7,8\}$, $\phi(x_4^-) = 6$ and $\phi(x_5^-) = z$. 
\end{proof}

\begin{theorem}
	If $G$ is a connected $2$-edge coloured cubic graph, then $\chi(G)\leq 10$.
\end{theorem}

\begin{proof}
	We proceed by contradicting the existence of a counter-example. Let $G$ be a $2$-edge coloured cubic graph such that $\chi(G) > 10$. 
	Certainly $G \not\to SP_9$.

	\emph{Claim 1: $G$ contains no induced copy of $K_4^{s+}$ and no induced copy of $K_4^{s-}$}.\\
	Assume otherwise. 
	Let $P_1, P_2, \dots, P_k$ be induced copies of $K_4^{s+}$ and let $N_1, N_2, \dots , N_\ell$ be induced copies of $K_4^{s-}$. 
	Let $P$ be the set of vertices contained in some copy of $K_4^{s+}$ and $N$ be the set of vertices contained in some copy of $K_4^{s-}$.
	By Lemma \ref{lem:subCubic} there is a homomorphism $\phi: G[V(G) \setminus P \cup N] \to SP_9$.
	By Lemma \ref{lem:extendPlus} $\phi$ can be extended to the vertices of $P_i$ for all $1 \leq i \leq k$.
	Similarly, by Lemma \ref{lem:extendNeg}, $\phi$ can be extended to the vertices of $N_j$ for all $1 \leq j \leq \ell$.
	Therefore $G \to SP_9^\dagger$.
	This is a contradiction as such a homomorphism yields is a $10$-colouring of $G$.

	\emph{Claim 2: No vertex of $G$ is positive or negative.}\\
	Assume otherwise. 
	Without loss of generality let $v$ be a vertex of $G$ with neighbours $u_1,u_2,u_3$ so that each of $vu_1, vu_2, vu_3$ is positive.
	By Claim 1, $G$ contains no copy of $K_4^{s+}$ or $K_4^{s-}$.
	Therefore by Lemma \ref{lem:subCubic},  there is a homomorphism $\phi: G-v \to SP_9$.
	We extend $\phi$ to be a $10$-colouring of $G$ by letting $\phi(v)=9$.

	\emph{Claim 3: $G$ contains no copy of $K_3$}.
	Assume otherwise. Let $u,v,w \in V(G)$ induce a copy of $K_3$ in $G$.
	Let $u^\prime\neq v,w$, $v^\prime\neq u,w$ and  $w^\prime \neq u,v$ be neighbours of $u,v$ and $w$ respectively in $G$.
	Without loss of generality, let $uv$ be negative.
	By Claim 1, $G$ contains no copy of $K_4^{s+}$ and no  copy of $K_4^{s-}$.
	Form $G^\prime$ from $G$ by removing the edge $uv$ and adding a vertex $z$, with positive edge $zu$ and negative edge $zv$.
	By Lemma \ref{lem:subCubic} there is a homomorphism $\phi: G^\prime \to SP_9$.
	Without loss of generality we may assume $\phi(u) = 0$ and $\phi(v) = 1$.
	Let $u^\prime\neq v,w$, $v^\prime\neq u,w$ and $w^\prime \neq u,v$ be neighbours of $u,v$ and $w$ respectively in $G$.
	Since $G$ does not admit a $10$-colouring it must be that $\phi(u^\prime) = 1$ and $\phi(v^\prime) = 0$.
	To see this, notice that if $\phi(u^\prime) \neq 1$, then letting $\phi(v) = 9$ yields a $10$-colouring of $G$.
	Consider now restricting $\phi$ to $G- \{u,v,w\}$.
	We can extend $\phi$ to include $u$ and $v$ so that $\phi(u) = 4$ and $\phi(v) = 6$ or $\phi(u) = 2$ and $\phi(v) = 3$.
	In particular, $\phi(u)$ and $\phi(v)$ can be chosen so that $\phi(u) \neq \phi(w^\prime)$ and $\phi(v) \neq \phi(w^\prime)$.
	Extending $\phi$ as such and letting $\phi(w) = 9$ yields a $10$-colouring of $G$, a contradiction.

	By Claim 2, we may partition the vertices of $G$ in to two sets $P$ and $N = V(G) \setminus P$ where vertices in $P$ are incident with exactly two positive edges and vertices in $N$ are incident with exactly two negative edges.
	
	\emph{Claim 4: There is no edge between a vertex of $P$ and a vertex of $N$}.\\
	Assume otherwise.
	Consider $u \in P$ and $v \in N$.
	Without loss of generality let $uv$ be a negative edge.
	Let $u_1\neq v $ and $u_2 \neq v$ be distinct  neighbours of $u$.
	Since $u \in P$ the edges $uu_1$ and $uu_2$ are both positive.
	Note that by Claim 4, $u_1$ and $u_2$ are not adjacent.
	Let $x \neq u$ be a neighbour of $v$ so that $vx$ is negative.
	Let $w$ be a neighbour of $v$ so that $vw$ is positive.
	Form $G^\star$ from $G$ by removing $v$ and adding a negative edge between $u_1$ and $u_2$.
	By Claim 3, we note that $G^\star$ contains no copy of $K_4^{s+}$ or $K_4^{s-}$; such $2$-edge-coloured graphs have at least three vertices incident with only positive or negative edges. By Claim 2, $G$ contains no such vertex and so $G^\star$ contains at most two.
	By Lemma \ref{lem:subCubic}, there is a homomorphism $\phi: G^\star \to SP_9$.
	By Lemma 4.5 we can extend $\phi$ to include $u$ so that $\phi(u) \neq \phi(w)$.
	Recolouring $v$ so that $\phi(v) = 9$ yields a $10$-colouring of $G$, a contradiction.
	
	\emph{Claim 5: $G$ does not exist.} As $G$ is connected and $\{N,P\}$ is a partition of $V(G)$, there is an edge between a vertex in $N$ and a vertex in $P$. This contradicts Claim 4.
\end{proof}

\begin{corollary}
	For $\mathcal{G}_3^c$, the family of connected $2$-edge coloured graphs, we have $8 \leq \mathcal{G}_3^c \leq 10$.
\end{corollary}

  \bibliographystyle{plain}
\bibliography{references}

\end{document}